\documentclass[11pt]{article}

\usepackage{amssymb} 
\usepackage{amsmath}     
\usepackage{amsthm}
\usepackage{todonotes}
\usepackage{mathtools}
\usepackage{a4wide}
\usepackage{graphicx}
\usepackage{hyperref}

\usepackage[ruled,vlined,linesnumbered]{algorithm2e}
\SetKw{KwNot}{not}
\SetKw{KwOR}{or}
\SetKw{KwAND}{and}
\SetKw{KwExitFor}{exit for-loop}
\SetKw{KwStep}{Step}

\usepackage{authblk}

\newcommand{\Zset}{\mathbb{Z}}

\newtheorem{thm}{Theorem}
\newtheorem{lem}{Lemma}

\newtheorem{prop}{Proposition}

\newenvironment{keyword}{\par{\noindent\bf Keywords:}}

\begin{document}

\title{Approximating the shortest path problem with scenarios}

\author[1]{Adam Kasperski}
\author[1]{Pawe{\l} Zieli\'nski\footnote{Corresponding author}}

\affil[1]{
Wroc{\l}aw  University of Science and Technology, Wroc{\l}aw, Poland\\
            \texttt{\{adam.kasperski,pawel.zielinski\}@pwr.edu.pl}}

\date{}

\maketitle

\begin{abstract}
This paper discusses the shortest path problem in a general directed
graph with $n$ nodes and $K$ cost scenarios (objectives). In order to choose a solution, the min-max criterion is applied. 
The min-max version of the problem is hard to approximate within 
$\Omega(\log^{1-\epsilon} K)$ for any $\epsilon>0$
unless NP\,$\subseteq \text{DTIME}(n^{\text{polylog} \,n})$  even for arc series-parallel graphs and
 within  $\Omega(\log n/\log\log n)$ unless NP\,$\subseteq \text{ZPTIME}(n^{\log\log n})$ for 
 acyclic graphs.
 The best approximation algorithm for the min-max shortest path problem in general graphs, known to date, has an approximation ratio of~$K$.  In this paper, 
an  $\widetilde{O}(\sqrt{n})$ flow LP-based approximation algorithm for min-max shortest path in general graphs
 is constructed. It is also shown that the approximation ratio obtained is close to an integrality gap of the corresponding flow LP relaxation.
 \end{abstract}

\begin{keyword}
  Combinatorial optimization;
  Approximation algorithm;
  Robust optimization;
  The shortest path problem;
  Uncertainty
  \end{keyword}

\section{Introduction}

In a \emph{combinatorial optimization problem}, denoted by~$\mathcal{P}$, we are given a finite set of \emph{elements} $E=\{e_1,\dots, e_m\}$ and a set of \emph{feasible solutions} $\Phi\subseteq 2^E$. In a deterministic case, each element $e\in E$ has some nonnegative integral cost~$c_e$, and we seek a feasible solution $X\in \Phi$ which minimizes a linear cost function $F(X)=\sum_{e\in X} c_e$. The above formulation encompasses a large class of problems. In this paper, we focus on  
a fundamental \emph{network problem}, where $E$ is the set of arcs of a given graph $G=(V,E)$, $|V|=n$,
and $\Phi$ contains the subsets of arcs 
forming $s$-$t$ paths in $G$. We thus consider the \emph{shortest path
problem} in a given graph $G$. We briefly denote it by~\textsc{SP}.

Let \emph{scenario set} $\mathcal{U}=\{\pmb{\xi}_1,\dots,\pmb{\xi}_K\}$ contain $K$ distinct cost scenarios, where
\emph{scenario} is a realization of the element costs,
 $\pmb{\xi}=(c_e^{\pmb{\xi}})_{e\in E}$ for $\pmb{\xi}\in\mathcal{U}\subset \Zset^{|E|}_{+}$.
  We  distinguish two cases, namely the \emph{bounded case}, when $K=O(1)$ 
  and the \emph{unbounded case}, when $K$ is a part of the input.
  The latter one is discussed in this paper. The  cost of a given solution~$X\in \Phi$ depends on scenario~$\pmb{\xi}$ and 
will be denoted by
$F(X,\pmb{\xi})=\sum_{e \in X} c_e^{\pmb{\xi}}$. In order to aggregate the cost vector 
$\pmb{F}(X,\mathcal{U})=(F(X,\pmb{\xi}_1),\dots,F(X,\pmb{\xi}_K))$ we use the maximum criterion, that is
the $\infty$-norm, $\lVert \pmb{F}(X,\mathcal{U})\rVert _{\infty}$.
We consider the following min-max version of~$\mathcal{P}$:
\begin{equation} 
\textsc{Min-Max}~\mathcal{P}:\; \min_{X\in \Phi}  \max_{\pmb{\xi}\in\mathcal{U}} F(X,\pmb{\xi})=
\min_{X\in \Phi} \lVert \pmb{F}(X,\mathcal{U})\rVert _{\infty}.
\label{minmaxp}
\end{equation}

 The \textsc{Min-Max}~$\mathcal{P}$ problem can be seen as a \emph{multi-objective problem}
 or a \emph{problem with extra constraints} or a \emph{$K$-budgeted one}
 (see, e.g.,~\cite{AFK02,GRSZ14,ER05,ER00, UT94,PY00}) 
 with the $\infty$-norm as an aggregation criterion. Each scenario defines a linear objective function, an extra constraint, or a budget constraint.
 In one interpretation, a scenario set $\mathcal{U}$ models the uncertainty of the element costs, and we seek a solution that hedges against the worst possible realization of the costs. This is a typical problem considered in
 \emph{robust optimization} (see, e.g.,~\cite{ABV09,BN09, KY97}).
 The \textsc{Min-Max}~$\mathcal{P}$
 is also  a special case of \emph{recoverable robust}~$\mathcal{P}$, in which
 a complete solution is chosen in the first stage, and  then, after a scenario from~$\mathcal{U}$ reveals,
 limited recovery actions are allowed in the second stage~\cite{B12,LLMS09}.
 The discrete scenario set $\mathcal{U}$ can be constructed from any probability distribution by performing a sampling (simulation)~\cite{KW94,RW91,W89}.  Obviously, the larger the set, the better the estimation of the uncertainty. So, the size of $\mathcal{U}$ in practical problems can be very large.
 This advocates for constructing  algorithms
that are applicable in the presence of uncertainty in the definition of the instance.

Unfortunately, the min-max versions of basic network optimization problems,  such as the minimum spanning tree,
the minimum assignment, the minimum $s$-$t$-cut, 
 are NP-hard, even if $K=2$~(see~\cite{ABV09,KZ16b} for surveys). 
In particular, this is also the case for the min-max version of the shortest path problem~\cite{YY98}.
When $K$ is a constant, the  \textsc{Min-Max~SP} problem admits 
a fully polynomial-time approximate scheme (FPTAS)~\cite{ABV07}.
 When $K$ is unbounded, the problem is strongly NP-hard and hard to approximate 
 within $\Omega(\log^{1-\epsilon} K)$ for any $\epsilon>0$
  unless NP\,$\subseteq \text{DTIME}(n^{\text{polylog} \,n})$~\cite{KZ09}, even for graphs with a 
 very simple structure, namely,  for \emph{arc series-parallel graphs} (see Section~\ref{sdsp} for the definition 
 of arc series-parallel graphs). 
 Recently, this bound has been slightly strengthened in~\cite{MOT2024,MOT2024a} for 
 acyclic graphs, namely under the assumption that NP\,$\not\subseteq \text{ZPTIME}(n^{\log\log n})$,
 the  \textsc{Min-Max~SP} problem is hard to approximate within 
 $\Omega(\log n/\log\log n)$.
 
 The best approximation algorithm for \textsc{Min-Max~SP} in general graphs, known to date,  has an approximation ratio of~$K$. It is based on a simple observation that an optimal solution for the average costs $\hat{c}_e=\frac{1}{K}\sum_{\pmb{\xi}\in \mathcal{U}} c_e^{\pmb{\xi}}$, $e\in E$, is at most $K$ times worse than the optimum (see~\cite{ABV09}).
In~\cite{BCFM17}, the authors claim that \textsc{Min-Max~SP} can be approximated within~$O(\log K)$, but
this claim turned out to be false (see~\cite{LXZ2024a,MOT2024a} for counterexamples).
Recently,  a randomized  $O(\log n \log K)$-approximation 
algorithm for \textsc{Min-Max~SP} in general digraphs, whose running time is quasi-polynomial,
has been designed~\cite{LXZ2024,LXZ2024a}.
The quasi-polynomial time of the approximation algorithm can be improved to polynomial for 
graphs with the bounded treewidth. This class of graphs includes the arc series-parallel graphs as a special case.
In~\cite{MOT2024,MOT2024a}, the authors study the $p$-norm shortest path problem of the form $\min_{X\in \Phi} \lVert \pmb{F}(X,\mathcal{U})\rVert _{p}$, where $p\in \Zset_{\geq 1}$. 
They proposed a randomized   $O(c p\log^{1-1/p}n)$-approximation algorithm for this problem 
 in general graphs, where~$c$ is a small constant, which runs in quasi-polynomial time.
 For arc series-parallel graphs, the quasi-polynomial running time of the above algorithm can be improved to polynomial time.
 Note that choosing $p=\lceil \log K \rceil$ in  the 
 $O(c p\log^{1-1/p}n)$-approximation algorithm, we get a randomized, quasi-polynomial time $O(c  \log n \log K)$-approximation algorithm
 for  \textsc{Min-Max~SP} in general graphs. This result is similar to the one in~\cite{LXZ2024,LXZ2024a}.

\textbf{Our results.} In this paper, we design a simple $\widetilde{O}(\sqrt{n})\footnote{This notation ignores logarithmic factors.}$-approximation algorithm (more precisely,
$O(\sqrt{n\log K/\log\log K})$-approximation one) for the min-max shortest path problem in general graphs. 
Contrary to the approximation algorithms proposed in~\cite{MOT2024,MOT2024a,LXZ2024a}, our algorithm is deterministic and runs in polynomial time.
It is based on an appropriate rounding of a solution to the flow LP relaxation of  \textsc{Min-Max~SP}. We will show that the approximation ratio obtained is very close to an integrality gap of this LP relaxation, which is at least  $\Omega(\sqrt{n})$,
even for arc series-parallel graphs.

This paper is organized as follows. In Section~\ref{slp}, we recall an LP formulation of \textsc{Min-Max}~$\mathcal{P}$, which leads to an LP relaxation for this problem. We also recall a formulation of the min-max version of the
 \emph{representatives selection problem}~\cite{DW13,DK12}. Our approximation algorithm, constructed in 
  Section~\ref{sdsp} will use the flow LP relaxation and some known approximation results on this selection problem. 
 In Section~\ref{sdsp}, we construct a deterministic $O(\sqrt{n\log K/\log\log K})$-approximation algorithm for \textsc{Min-Max~SP}
 in general graphs. 
 We also show that an integrality gap of the  flow LP relaxation is at least $\Omega(\sqrt n)$,
even for arc series-parallel graphs.

\section{LP relaxation and representatives selection problem}
\label{slp}

The min-max problem~(\ref{minmaxp})
can be alternatively stated as the following integer program:
\begin{align}
		OPT=\min  & \;C & \label{mipminmax0}\\
		\text{s.t. }	&  \sum_{e\in E} c_e^{\pmb{\xi}}x_e\leq C & \forall \pmb{\xi}\in \mathcal{U}, \label{mipminmax1}\\
		&  \pmb{x}\in \mathcal{X}, &\label{mipminmax2}\\
		&  \pmb{x}\in \{0,1\}^{|E|}, \label{mipminmax3}&
\end{align}
where (\ref{mipminmax2}) and (\ref{mipminmax3}) describe the set of feasible solutions~$\Phi$,
 $\mathcal{X}$ is given by a system of linear constraints involving~$\pmb{x}=(x_{e})_{e\in E}$ and
 $\pmb{x}$ is a characteristic vector of a feasible solution $X\in\Phi$.
 
Fix a parameter~$C>0$ and let $E(C)=\{e\in E: \, (\forall \pmb{\xi}\in \mathcal{U}) (c_e^{\pmb{\xi}}\leq C)\} \subseteq E$. 
Consider the following family of linear programs:
\begin{align}
\mathcal{LP}(C):&\sum_{e\in E} c_e^{\pmb{\xi}}x_e\leq C &  \forall \pmb{\xi}\in \mathcal{U},\label{sc}\\
                          & \pmb{x}\in \mathcal{X},& \label{hc}\\
                          &x_e \in [0,1] & \forall e\in E(C), \label{hc1} \\
			&x_e=0 & \forall e\notin E(C).  \label{hc2}
\end{align}
Minimizing $C$ subject to~(\ref{sc})-(\ref{hc2}) we obtain an \emph{LP relaxation} of (\ref{mipminmax0})-(\ref{mipminmax3}).
Let  $C^*$ denote  the smallest value of the parameter~$C$,
for which $\mathcal{LP}(C)$ is feasible and let $\pmb{x}^*$ be a feasible solution to~$\mathcal{LP}(C^*)$.
The value of $C^*$ is a lower bound on  $OPT$
and can be determined in polynomial time by using the following algorithm. Let $C_{\max}$ be an upper bound on $C^*$ (we can use, for example, $C_{\max}=\max_{\pmb{\zeta}\in \mathcal{U}}\sum_{e\in E} c^{\pmb{\zeta}}_e$). Using a binary search, we first determine in $O(T(K,|E|)\log C_{\max})$ time the smallest integer value $C'\in [0,C_{\max}]$ for which $\mathcal{LP}(C')$ is feasible,
where $T(K,|E|)$  is the time required to solve the linear program.
Then $C^*\in (C'-1,C']$ and $C^*$ can be computed by minimizing $C$ subject   to~(\ref{sc})-(\ref{hc2}) for the fixed set $E(C'-1)$. If this linear program is infeasible, then $C^*=C'$. The correctness of this method follows from the assumption that 
$\mathcal{U}\subset \Zset^{|E|}_{+}$.
 Our approximation algorithm, constructed in the next section, will be based on appropriately rounding the solution $\pmb{x}^*$.

Let us recall a special case of $\textsc{Min-Max}~\mathcal{P}$, which will be used later in this paper. 
Namely, 
 the \emph{min-max representatives selection} problem
  (\textsc{Min-Max RS} for short)~\cite{DW13,DK12,KKZ15}.
We are given a set $E=\{e_1,\dots, e_m\}$  of~$m$ tools and
  $E$ is partitioned into $p$ disjoint sets $E_1,\ldots,E_p$, where $|E_i|=r_i$ and $m=\sum_{i\in [p]}r_i$. 
  We will denote by $[p]$ the set $\{1,2,\ldots,p\}$.
  Define $\Phi=\{X\subseteq E\,:\,  (\forall i\in[p]) (|X\cap E_i|=1)\}$, so we wish to choose a subset $X\subseteq E$ of the tools that contains exactly one tool from each set $E_i$ to minimize the maximum cost over $\mathcal{U}$.
The set
  $\mathcal{X}$ in~(\ref{mipminmax2})  and~(\ref{hc}), for \textsc{Min-Max RS}, is then described by~$p$ constraints of the form
  \begin{equation}
  \sum_{e\in E_i} x_e=1\;\;\forall i\in[p]. \label{Xrs}
  \end{equation}
  It is worth pointing out that
the problem can be solved trivially for the case~$K=1$.
 Unfortunately, for unbounded scenario set $\mathcal{U}$,  \textsc{Min-Max RS} is strongly NP-hard and
 hard to approximate
 within $\Omega(\log^{1-\epsilon}K)$ for any $\epsilon>0$
unless NP\,$\subseteq \text{DTIME}(n^{\text{polylog} \,n})$~\cite{KKZ15}. In the next part of the paper, we will use the following result:
  \begin{thm}\cite{KKZ15}
\label{tsi}
 Let $\pmb{x}$ be a fractional feasible solution to the linear program $\mathcal{LP}(C)$ with~$\mathcal{X}$ of the form~(\ref{Xrs}) and 
 assume that $\max_{\pmb{\xi}\in \mathcal{U}}\sum_{e\in E} c_e^{\pmb{\xi}} x_e=C$. 
 Then there is an algorithm which transforms $\pmb{x}$, in $O(Km\log m)$ time, into a
 solution~$X$ for \textsc{Min-Max RS} that has cost $F(X,\pmb{\xi})=O(C \log K/\log\log K)$ for every 
 $\pmb{\xi}\in \mathcal{U}$.
\end{thm}  
Theorem~\ref{tsi} immediately leads to the LP-based $O(\log K/\log\log K)$ approximation algorithm for \textsc{Min-Max RS}. It is enough to choose $C=C'$,  where~$C'$ is
the smallest value of the parameter~$C$,
for which  the linear program $\mathcal{LP}(C)$ with~(\ref{Xrs}) is feasible
and use the fact that $C'$ is a lower bound on~$OPT$.

\section{Min-max shortest path}
\label{sdsp}

Let $G=(V,E)$ be a directed graph with two specified nodes $s,t\in V$, where $|V|=n$ and $|E|=m$. The set of feasible solutions $\Phi$ contains all $s$-$t$ paths in $G$. Each scenario $\pmb{\xi}\in \mathcal{U}$ represents a realization of the arc costs. In this section, we construct an LP-based $O(\sqrt{n\log K/\log\log K})$-approximation algorithm for the min-max version of the shortest path problem when $|\mathcal{U}|=K$ is unbounded.
For this problem, the set $\mathcal{X}$ in~(\ref{mipminmax2})  and~(\ref{hc}) is described by the following 
\emph{flow constraints}
 (the mass
 balance constraints):
 \begin{align}
  \sum_{e\in \delta^{\mathrm{out}}(v)}x_e&=\sum_{e\in \delta^{\mathrm{in}}(v)}x_e&
   \forall v\in V\setminus \{s,t\}, \label{xsp1}\\
  \sum_{e\in \delta^{\mathrm{out}}(s)}x_e&=\sum_{e\in \delta^{\mathrm{in}}(t)}x_e=1,& 
  \label{xsp2}
 \end{align} 
where $\delta^{\mathrm{out}}(v)$ and $\delta^{\mathrm{in}}(v)$ are the sets of outgoing and incoming
arcs, respectively, from $v\in V$.

We now show that the linear relaxation  $\mathcal{LP}(C)$ with the  flow constraints~(\ref{xsp1})-(\ref{xsp2})
has an integrality gap of at least $\Omega(\sqrt{n})$,
even for graphs with a 
 very simple structure, namely,  for \emph{arc series-parallel graphs}.
An arc series-parallel graph (ASP) is
recursively defined as follows (see, e.g.,~\cite{VTL82}). A graph
consisting of two nodes joined by a single arc is ASP. If $G_1$
and $G_2$ are ASP, so are the graphs constructed by each of
the following two operations:
\begin{itemize}
\item \emph{parallel composition}: identify the
source~$s_1$ of $G_1$ with the source~$s_2$ of $G_2$ and the sink~$t_1$ of $G_1$
with the sink~$t_2$ of~$G_2$,
\item \emph{series composition}: identify the sink~$t_1$ of $G_1$ with the source~$s_2$ of~$G_2$.
\end{itemize}
\begin{prop}
The linear relaxation  $\mathcal{LP}(C)$ with the flow constraints~(\ref{xsp1})-(\ref{xsp2})
has an integrality gap of at least $\Omega(\sqrt{n})$, even for arc series-parallel graphs.
\label{gapsp}
\end{prop}
\begin{proof}
See~Appendix~\ref{dod}.
\end{proof}

\noindent\emph{Remark.} It is worth noting that $s$-$t$~paths in $G$ can be also expressed by the $s$-$t$~\emph{cut constraints}
(see, e.g.,~\cite{WS10}).
Let us recall that an \emph{$s$-$t$~cut} in $G$ is a partition of $V$ into the subsets~$S$ and $\overline{S}=V\setminus S$ such that $s\in S$ and $t\in \overline{S}$. 
  We  refer to an arc $(u,v)$ with $u\in S$ and $v\in \overline{S}$ as
a \emph{forward arc} of the $s$-$t$ cut.
  The \emph{cut-set} of an $s$-$t$~cut, denoted by $(S,\overline{S})$, is the set of forward arcs  of the $s$-$t$ cut,
  i.e. the set  $\{(u,v)\in E: u\in S, v\in \overline{S}\}$.
 The $s$-$t$~\emph{cut constraints} have the following form:
\begin{equation}
  \sum_{e\in (S,\overline{S})} x_e\geq 1
   \text{ for all $(S,\overline{S})$ in~$G$}. \label{ccxsp}\\
 \end{equation} 
If we use the $s$-$t$~cut constraints~(\ref{ccxsp}) for expressing  $s$-$t$~paths  
 in 
the linear relaxation~$\mathcal{LP}(C)$, instead of the flow ones~(\ref{xsp1})-(\ref{xsp2}), then an
analysis similar to that in the proof of Proposition~\ref{gapsp}  shows that an integrality gap also remains at least $\Omega(\sqrt{n})$, even for arc series-parallel graphs.

We are ready to present 
the approximation algorithm for \textsc{Min-Max~SP}. It
is shown in the form of Algorithm~\ref{algmmsp}. In the following, we will describe all its steps.
In Step~\ref{spwl1}, we compute a fractional solution $\pmb{x}^*$ to $\mathcal{LP}(C^*)$ with the constraints~(\ref{xsp1})-(\ref{xsp2}),
 which is an $s$-$t$ arc flow in~$G$.
In Steps~\ref{spwl2} and~\ref{spwl3}, we preprocess  graph~$G$.
Let us recall the \emph{flow decomposition theorem} (see, e.g.,~\cite{AMO93}).
\begin{thm}[Flow decomposition]
\label{tfd}
Every nonnegative $s$-$t$ arc flow
$\pmb{x}$ can be represented as 
at most $m$ directed cycles and $s$-$t$ paths with nonzero flows (although not necessarily uniquely).
The value of the flow from $s$ to $t$ is equal to the sum of the flows along the  $s$-$t$ paths.
\end{thm}
Theorem~\ref{tfd} shows that the $s$-$t$ arc flow~$\pmb{x}^*$ can be decomposed into directed cycles and $s$-$t$ paths with nonzero flows.
We first remove from~$G$ every arc~$e\in E$ with~$x^*_e=0$ (Step~\ref{spwl2}). The resulting graph with $x^*_e>0$ for each~$e\in E$, must be connected, i.e. it must have an $s$-$t$ path.
Furthermore, to reduce the problem size, one may perform \emph{series arc reductions} (the  inverse  operation
to the series composition) in~$G$,
i.e. the operations which replace two series arcs~$f,g\in E$ by the single arc~$w$ with the cost $c^{\pmb{\xi}}_w=c^{\pmb{\xi}}_{f}+c^{\pmb{\xi}}_{g}$ under 
 $\pmb{\xi}\in \mathcal{U}$ and $x^*_w=x^*_f$. 
Finally, we convert $\pmb{x}^*>\pmb{0}$ into a feasible solution~$\pmb{\hat{x}}$ to $\mathcal{LP}(C^*)$ by removing directed cycles, 
such that the resulting graph with~$\hat{x}_e>0$  for every arc~$e\in E$ is acyclic (Step~\ref{spwl3}). A procedure is described in
 the proof of  following lemma.
\begin{lem}
Let $\pmb{x}^*>\pmb{0}$ be a feasible   solution to  $\mathcal{LP}(C^*)$ with the flow constraints~(\ref{xsp1})-(\ref{xsp2}).
Then there exists a feasible   solution~$\pmb{\hat{x}}$ to  $\mathcal{LP}(C^*)$
with~(\ref{xsp1})-(\ref{xsp2})
 such that
the graph induced  by~$\pmb{\hat{x}}>\pmb{0}$ is acyclic.
\label{lpresp}
\end{lem}
\begin{proof}
Consider $\pmb{x}^*>\pmb{0}$ and the graph~$G$ induced  by~$\pmb{x}^*>\pmb{0}$.
Theorem~\ref{tfd} implies that we can decompose the $s$-$t$ arc flow~$\pmb{x}^*$  into
 nonzero flows along   $s$-$t$  paths and cycles.
If each arc cost under any scenario is positive,  then $G$ is also acyclic, and we are done.
For nonnegative arc costs, $G$ may have a cycle.  Consider such a cycle and denote it by~$\mathcal{W}$.
Let $\mathcal{U}^{=}$ be the subset of $\mathcal{U}$ such that 
constraints~(\ref{sc}) are tight for $\pmb{x}^*$ and $C^*$.
We claim that there is at least one scenario $\pmb{\xi'}\in \mathcal{U}^{=}$ such that
$\sum_{e\in \mathcal{W}}  c_e^{\pmb{\xi'}}=0$. 
Suppose, contrary to our claim, that $\sum_{e\in \mathcal{W}}  c_e^{\pmb{\xi}}>0$
for every  $\pmb{\xi}\in \mathcal{U}^{=}$. 
Since for each $e\in \mathcal{W}$, $x^*_{e}>0$,
we can decrease the flow~$\pmb{x}^*$ on cycle~$\mathcal{C}$ by $\Delta=\min\{x^*_e\; :\; e\in \mathcal{C}\}$
and, in consequence,  decrease  the cost of~$\pmb{x}^*$ under every $\pmb{\xi}\in \mathcal{U}^{=}$,
which contradicts the optimality of~$C^*$.
Accordingly, we can decrease the flow~$\pmb{x}^*$ on the cycle~$\mathcal{W}$ by $\Delta$ without affecting~$C^*$.
The resulting solution is still feasible to (\ref{hc})-(\ref{hc2}), its maximum cost over~$\mathcal{U}$ is equal to~$C^*$ and
at least one arc from~$\mathcal{W}$ has zero flow. Thus, the graph induced by this new solution does not contain
the  cycle~$\mathcal{W}$. Applying the above procedure to all cycles in~$G$ one can convert $\pmb{x}^*$ into
a feasible   solution~$\pmb{\hat{x}}$ to  $\mathcal{LP}(C^*)$
such that the  new graph~$\hat{G}$  with~$\hat{x}_e>0$  for every arc~$e$ is acyclic.
Now the $s$-$t$ arc flow~$\pmb{\hat{x}}>\pmb{0}$
can be decomposed only into
 nonzero flows along  $s$-$t$ paths. This completes the proof.
\end{proof}
 The feasible   solution~$\pmb{\hat{x}}$ from Lemma~\ref{lpresp} can be determined in $O(mn)$ time
  using the flow decomposition 
 algorithm (see, e.g.,~\cite{AMO93}).

From now on, we will assume that the graph $G$ is induced by~$\pmb{\hat{x}}>\pmb{0}$ is acyclic.
 Let us assign an
arc length~$l_e\in \{0,1\}$ to each arc of $G$.  If $l_e=0$, then the arc is called \emph{selected}; otherwise, $l_e=1$, it is called \emph{not selected}. Initially, each arc is not selected, so the length~$l_e=1$ for each $e\in E$ (Step~\ref{spwl4}). We will carefully mark some arcs as selected
by assigning them the zero arc length
 during the course of the algorithm. The selected arcs form connected components in $G$, where each connected component of~$G$ has 
 a total length equal to~0 with respect to $l_e$ (initially, the nodes of $G$ form $n$ connected components because no arc is selected).   
 The following lemma will be needed to construct
 a feasible solution to the relaxation
 (\ref{sc})-(\ref{hc2}) of \textsc{Min-Max RS} with $\mathcal{X}$ of the form~(\ref{Xrs}).
\begin{lem}
Let  $L_P$ be  the length of a shortest $s$-$t$~path~$P$ in the acyclic graph~$G$,
induced by~$\pmb{\hat{x}}>\pmb{0}$,
with respect to 
the arc lengths~$l_e \in \{0,1\}$, $e\in E$. If $L_P\geq 1$,
 then there are $L_P$ arc disjoint cut-sets: $(S_1,\overline{S}_1),\ldots, (S_{L_P},\overline{S}_{L_P})$ in $G$, 
 which do not contain any selected arc.
 Moreover, $\sum_{e\in (S_i,\overline{S}_i)}\hat{x}_e\geq 1$ for each $i\in [L_P]$.
 \label{lspexst}
\end{lem}
\begin{proof}
Let $d(v)$ be the shortest path distance from $s$ to $v$, $v\in V$, with respect to $l_e$, $e\in E$.
The shortest path distances satisfy an \emph{optimality condition}  (see, e.g.,~\cite{AMO93}), i.e.
$d(u)+l_e\geq d(v)$, for every~$e=(u,v)\in E$.
Of course, $d(s)=0$ and $d(t)=L_P$.  Using the shortest path distances~$d(v)$, $v\in V$, computed, one can determine $L_P$ cut-sets in the following way:
\begin{equation}
 (S_i,\overline{S}_i)=\{ (u,v)\in E \,:\,  d(u)=i-1, d(v)=i \}, \;\;i\in [L_P].\label{llcutset}
 \end{equation}
 We can see at once that all the $L_P$ cut-sets
 are nonempty, since $L_P\geq 1$, and by the definition
they are arc disjoint. 
If an arc~$e=(u,v)$ is such that $l_e=0$ ($e$ is  selected), then by the optimality condition we get
$d(u)\geq d(v)$ and, in consequence, $e\not\in (S_i,\overline{S}_i)$ for every $i\in [L_P]$.
Thus, these cut-sets do not contain any selected arc.

The value of the $s$-$t$ arc flow~$\pmb{\hat{x}}$ is equal to~1. By 
Theorem~\ref{tfd},  $\pmb{\hat{x}}$ can be decomposed only into $s$-$t$ paths with nonzero flows ($G$ is acyclic) and
 the sum of the flows along these  $s$-$t$ paths is equal to~1.
Consider any cut-set $(S_i,\overline{S}_i)$,  $i\in [L_P]$.
 Since each $s$-$t$ path contains at least one arc in $(S_i,\overline{S}_i)$, 
  we have $\sum_{e\in (S_i,\overline{S}_i)}\hat{x}_e\geq 1$.
\end{proof}
Notice that the cut-sets $(S_1,\overline{S}_1),\ldots, (S_{L_P},\overline{S}_{L_P})$, in Lemma~\ref{lspexst}, can be determined in $O(m)$ time.
An example with determined cut-sets is shown in Figure~\ref{fig1}.
\begin{figure}[ht]
\centering
\includegraphics[width=\textwidth]{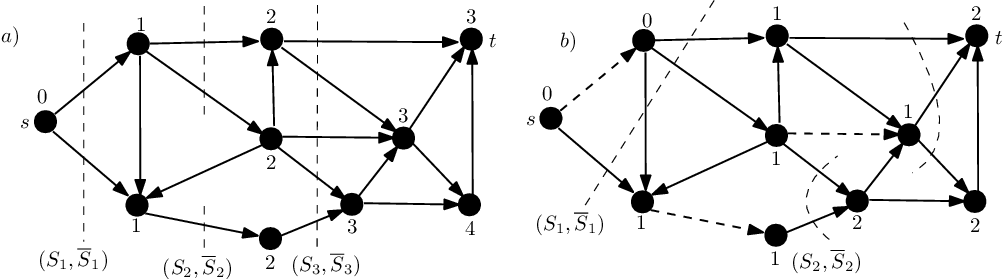}
\caption{(a) The cut-sets determined in  the first round: $(S_1,\overline{S}_1)$, $(S_2,\overline{S}_2)$, 
$(S_3,\overline{S}_3)$ and $L_P=3$,
and  (b) the cut-sets determined in the second round:  $(S_1,\overline{S}_1)$, $(S_2,\overline{S}_2)$ and $L_P=2$. The dashed arcs are selected and have $l_e=0$, and the solid arcs are not selected and have $l_e=1$. The numbers assigned to the nodes are the shortest path distances~$d(v)$, $v\in V$.}\label{fig1}
\end{figure}

The algorithm performs $\hat{k}$ rounds (see Steps~\ref{rstart}-\ref{rend}). In the $k$th round, $k\in [\hat{k}]$, we compute a shortest $s$-$t$~path $P$ with respect to the arc lengths~$l_e$, $e\in E$
(see Steps~\ref{spwl6} and~\ref{rend}). If $L_P$ is less than the prescribed value 
$\ell$ (see Step~\ref{spwl}), then we terminate and output $P$. Otherwise, we find $L_P$ arc disjoint cut-sets $(S_1,\overline{S}_1),\ldots, (S_{L_P},\overline{S}_{L_P})$, described in Lemma~\ref{lspexst} (see also~(\ref{llcutset}) and Figure~\ref{fig1}).  These cut-sets form an instance of the \textsc{Min-Max RS} problem with $K$ scenarios induced by the scenarios in $\mathcal{U}$.
 Since $\sum_{e\in (S_i,\overline{S}_i)}\hat{x}_e\geq 1$ for each $i\in [L_P]$,
 we have to perform the normalization in Step~\ref{norm}.
 As a result,   we get a feasible solution $\pmb{x'}$ to the relaxation
 (\ref{sc})-(\ref{hc2}) of \textsc{Min-Max RS} with $\mathcal{X}$ of the form~(\ref{Xrs}), $p=L_P$
 and $E_i=(S_i,\overline{S}_i)$, $i\in [L_P]$, i.e.
to the linear program $\mathcal{LP}(C')$, 
where $C'=\max_{\pmb{\xi}\in \mathcal{U}} \sum_{i\in[L_p]}\sum_{e\in E_i} c_e^{\pmb{\xi}} x'_e$.
We now use Theorem~\ref{tsi} to pick a set of arcs $X$, $|X|=L_P$, that contains exactly one arc from each $(S_i,\overline{S}_i)$
(see Step~\ref{spwl11}).  The following lemma describes the cost of $X$.
\begin{lem}
 The value of $\max_{\pmb{\xi}\in \mathcal{U}} \sum_{e\in X} c_e^{\pmb{\xi}}$ is $O(C^*\log K/\log\log K)$.
 \label{lspcost}
\end{lem} 
\begin{proof}
 The following inequalities:
\begin{align*}
C^*&=\max_{\pmb{\xi}\in \mathcal{U}} \sum_{e\in E} c_e^{\pmb{\xi}} \hat{x}_e \geq
 \max_{\pmb{\xi}\in \mathcal{U}} \sum_{i\in[L_P]}\sum_{e\in E_i} c_e^{\pmb{\xi}} \hat{x}_e\geq
 \max_{\pmb{\xi}\in \mathcal{U}} \sum_{i\in[L_P]}\sum_{e\in E_i} c_e^{\pmb{\xi}}
  \frac{\hat{x}_e}{ \displaystyle\sum_{e\in E_i}\hat{x}_e}\\
&= \max_{\pmb{\xi}\in \mathcal{U}} \sum_{i\in[L_P]}\sum_{e\in E_i} c_e^{\pmb{\xi}} x'_e=C'
\end{align*}
hold. The first inequality follows from $\bigcup_{i\in [L_P]} E_i\subseteq E$ and
the second one from the fact that $\sum_{e\in E_i}\hat{x}_e\geq 1$, $i\in [L_P]$ (see Lemma~\ref{lspexst}).
Since $\pmb{x'}$ is a feasible solution to  $\mathcal{LP}(C')$ with $\mathcal{X}$ of the form~(\ref{Xrs})
(with $p=L_P$),
Theorem~\ref{tsi} combined with $C'\leq C^*$  gives the assertion of the lemma.
\end{proof}
Note that each selected arc~$e\in X$  (every with~$l_e=1$ by~(\ref{llcutset}))   
  merges two connected components. Each connected component  
  consists of one node or nodes connected by already selected arcs with zero lengths.
  Hence, and from $|X|=L_P$, we get that the number of connected components is reduced by~$L_P$. Then,
  for each selected arc  $e\in X$, we   set $l_e=0$, marking it as selected (see Step~\ref{spwl12}).
  
  \begin{algorithm}
 \SetKwInOut{Input}{Input}
 \SetKwInOut{Output}{Output}
  \BlankLine
 \Input{A directed graph $G=(V,E)$, $|V|=n$, $|E|=m$, two specified nodes $s,t\in V$, a cost scenario set 
 $\mathcal{U}=\{\pmb{\xi}_1,\dots,\pmb{\xi}_K\}$, $\pmb{\xi}=(c_e^{\pmb{\xi}})_{e\in E}\in \mathcal{U}.$}
  \Output{An $s$-$t$ path~$P$ in~$G$ -- an approximate solution  for \textsc{Min-Max~SP}.}
  \BlankLine
  Find $C^*$ and a feasible solution $\pmb{x}^*=(x^*_e)_{e\in E}$ 
  to $\mathcal{LP}(C^*)$ with the flow  constraints~(\ref{xsp1})-(\ref{xsp2})\; \label{spwl1}
\lForEach{$e\in E$ with $x^*_e=0$}{$E\leftarrow E\setminus \{e\}$}  \label{spwl2}
 Convert $\pmb{x}^*$ into solution~$\pmb{\hat{x}}$, feasible to $\mathcal{LP}(C^*)$, i.e.
  perform series arc reductions and remove cycles in the graph~$G$ induced by~$\pmb{x}^*$
  (Lemma~\ref{lpresp})\; \label{spwl3}
\tcc{Now $G$ with $\hat{x}_e>0$ for every arc~$e$ is acyclic}
\lForEach(\tcp*[f]{Set arc lengths}){$e\in E$}{$l_e\leftarrow 1$} \label{spwl4}
$\ell \leftarrow\sqrt{n\log K/\log\log K}$ \tcp*{a prescribed value} \label{spwl}
Find a shortest $s$-$t$ path~$P$ in $G$ with arc lengths~$l_e$, $e\in E$;
let $L_P$ be the length of this path\; \label{spwl6}
 \While(\tcp*[f]{round $k$}){$L_P> \ell$}{
\label{rstart}
Find $L_P$ arc disjoint cut-sets $(S_i,\overline{S}_i)$ in $G$, $i\in [L_P]$ described in Lemma~\ref{lspexst}\;\label{spwl8}
\ForEach{$i\in [L_P]$}{
\lForEach{$e\in (S_i,\overline{S}_i)$}{$x'_e\leftarrow \hat{x}_e/(\sum_{e\in (S_i,\overline{S}_i)}\hat{x}_e)$} \label{norm}
}
Use Theorem~\ref{tsi} to transform $x'_e$, $e\in \bigcup_{i\in [L_P]} E_i$, where $E_i=(S_i,\overline{S}_i)$,
into a
 solution~$X$, $|X|=L_P$, for \textsc{Min-Max RS}\;  \label{spwl11}
 \lForEach(\tcp*[f]{$e$ is selected}){$e\in X$}{$l_e\leftarrow 0$}  \label{spwl12}
 Find a shortest $s$-$t$ path $P$ in $G$ with arc lengths~$l_e$, $e\in E$; let $L_P$ be the length of~$P$\;\label{rend}
} 
\Return{$P$}
  \caption{An approximation algorithm  for \textsc{Min-Max~SP}}
 \label{algmmsp}
\end{algorithm}

 We now prove the following lemma needed to analyze the performance of Algorithm~\ref{algmmsp}.
\begin{lem}
Let $\ell$ be a prescribed value in Algorithm~\ref{algmmsp} (Step~\ref{spwl}). 
Then Algorithm~\ref{algmmsp}   performs at most $n/\ell$ rounds (Steps~\ref{rstart}-\ref{rend})
and
returns an $s$-$t$~path~$P$ with the cost $\max_{\mathcal{\xi\in U}} \sum_{e \in P} c_e^{\pmb{\xi}}$
  at most $O\left(\frac{n\log K}{\ell\log\log K} +\ell\right)C^*$.
  \label{lspph}
\end{lem}
\begin{proof}
We will denote by~$\hat{k}$ the number of rounds performed by the algorithm (Steps~\ref{rstart}-\ref{rend}).
Let $\mathcal{C}^k$ be the number of connected components in $G$ (merged by selected arcs) and $L_P^{k}$ be the length of 
a shortest path from~$s$ to~$t$ in~$G$
  with respect to~$l_e\in\{0,1\}$, $e\in E$, at the beginning of the $k$th round, $k\in[\hat{k}]$.
  In the $k$th round, we choose $X^k$, $|X^k|=L_P^k$, and set $l_e=0$ for each $e\in X^k$, which reduces the number of connected components by $L_P^k$.
 In consequence, the equalities
  \begin{equation}
    \mathcal{C}^{k+1}=\mathcal{C}^k-L_P^k, \;\mathcal{C}^1=n,\; k\in  [\hat{k}],
    \label{lsppheq1}
  \end{equation}
  hold.
  Observe that $L^{\hat{k}+1}_P\leq \ell$,  $\ell<L^k_P$, $k\in  [\hat{k}]$ (see the condition~\ref{rstart}). Furthermore,  $L^{\hat{k}+1}_P<\mathcal{C}^{\hat{k}+1}$. Indeed, each connected component in the $\hat{k}$th round consists of the arcs with length $l_e=0$. A shortest path with respect to $l_e\in \{0,1\}$ connects at most~$\mathcal{C}^{\hat{k}+1}$ components using the arcs with length $l_e=1$. So, the length of this path $L^{\hat{k}+1}_P$ is at most~$\mathcal{C}^{\hat{k}+1}-1$, which implies $L^{\hat{k}+1}_P<\mathcal{C}^{\hat{k}+1}$.
  Hence and from (\ref{lsppheq1}), we obtain
  \[
   L^{\hat{k}+1}_P< \mathcal{C}^{\hat{k}+1}=\mathcal{C}^1-\sum_{k\in [\hat{k}]} L^k_P < n-\hat{k} \ell,
  \]
  which gives the following bound on the number of rounds performed:
  \[
  \hat{k}<(n- L^{\hat{k}+1}_P)/\ell\leq n/\ell.
\]

  Consider the cost of the path $P$ returned.
  Since  each arc~$e\in P$
  is such that $e\in E(C^*)$ and
  the number of not selected arcs on this path is at most $\ell$,
  the maximum total cost of not selected arcs on $P$ is at most $\ell C^*$.
  On the other hand, the total cost of all selected arcs is at most $O(\frac{n \log K}{\ell \log \log K} C^*)$, which results from applying 
  at most
  $n/\ell$ times Lemma~\ref{lspcost}. So, the maximum cost of $P$ over all scenarios in $\mathcal{U}$ is $O\left(\frac{n\log K}{\ell\log\log K} +\ell\right)C^*$.
\end{proof}
The best ratio in $O\left(\frac{n\log K}{\ell\log\log K} +\ell\right)$
can be achieved by choosing $\ell=\sqrt{n\log K/\log\log K}$, which is a prescribed bound on the number of
not selected arcs on $P$ in  Algorithm~\ref{algmmsp} (see Step~\ref{spwl}). Lemma~\ref{lspph}  
and $C^*\leq OPT$  yield  an approximation factor of~ $O(\sqrt{n\log K/\log\log K})$.

We have yet to show that the running time of Algorithm~\ref{algmmsp}  is polynomial. 
Let us first analyze Steps~\ref{spwl1}-\ref{spwl6}.
 The value of~$C^*$ and 
corresponding  feasible solution $\pmb{x}^*$ 
  to $\mathcal{LP}(C^*)$  can be found in  $O(T(K,m,n)\log C_{\max})$ time, where $C_{\max}=n\max_{\pmb{\xi}\in\mathcal{U}, e\in E} \{c_e^{\pmb{\xi}}\}$ and  $T(K,m,n)$  is the time required to solve the linear program with constraints $\mathcal{LP}(C)$ with~(\ref{hc}) of the form~(\ref{xsp1})-(\ref{xsp2}).
The time $T(K,m,n)$ is polynomial (see, e.g.,~\cite{JSWZ21}). 
Steps~\ref{spwl2} and \ref{spwl4} require $O(m)$ time.
Converting $\pmb{x}^*$ into solution~$\pmb{\hat{x}}$ in Step~\ref{spwl3}
can be done in $O(mn)$ time
 by using the flow decomposition algorithm (see, e.g.,~\cite{AMO93}).
 Since $G$ is acyclic, a shortest $s$-$t$ path~$P$ in $G$ (Step~\ref{spwl6} and also~\ref{rend})
 can be found in $O(m)$ time (see, e.g.,~\cite{AMO93}).
 Therefore, Steps~\ref{spwl1}-\ref{spwl6} require $O(mn+T(K,m,n)\log C_{\max})$ time.
 Consider the round~$k$ (Steps~\ref{rstart}-\ref{rend}).
 Steps~\ref{spwl8}, \ref{norm} and~\ref{rend} require~$O(m)$ time
 (in  Step~\ref{spwl8}, we use already computed shortest path distances).
 By  Theorem~\ref{tsi}, Step~\ref{spwl11} requires in $O(Km\log m)$ time.
 The number of rounds is bounded by $\frac{n}{\ell}=\frac{\sqrt{n}}{\sqrt{\log K/\log\log K}}$ (see Lemma~\ref{lspph}).
 Thus,  executing  the  while loop (Steps~\ref{rstart}-\ref{rend}) takes $O\left(\frac{Km\sqrt{n}\log m}{\sqrt{\log K/\log\log K}}\right)$ time.
 Hence, the total running time of Algorithm~\ref{algmmsp}  is polynomial.
 
 Accordingly, we have thus proved the main result of this paper.
\begin{thm}
There is an $O(\sqrt{n\log K/\log\log K})$-approximation algorithm for \textsc{Min-Max~SP}
in a general directed graph.
\end{thm}

\section{Conclusions}
\label{scon}

There is still an open question concerning
the \textsc{Min-Max~SP} problem.
For this problem, there exists  an $O(\sqrt{n\log K/\log\log K})$-approximation algorithm
($\widetilde{O}(\sqrt{n})$-approximation algorithm) for general graphs,
designed in this paper,
 $\Omega(\log^{1-\epsilon} K)$ lower bound  on the  approximability of  the problem
 for arc series-parallel graphs~\cite{KZ09},  and    $\Omega(\log n/\log\log n)$ lower bound on the  approximability for acyclic graphs~\cite{MOT2024}.  We have also shown that an integrality gap of the flow LP relaxation is very close to the approximation ratio of the constructed algorithm. This suggests that better approximation algorithms based on rounding the solutions to the flow LP relaxation may not be possible.
 Closing the gap between the positive and negative approximation results for \textsc{Min-Max SP} may require other techniques and is an interesting subject for further
research. The first step has been taken in~\cite{LXZ2024} and~\cite{MOT2024},
where randomized polylogarithmic approximation algorithms for this problem,  which run in
a quasi-polynomial time, have been constructed.

\subsubsection*{Acknowledgements}
The authors were supported by
 the National Science Centre, Poland, grant 2022/45/B/HS4/00355.


\begin{thebibliography}{10}

\bibitem{AMO93}
R.~K. Ahuja, T.~L. Magnanti, and J.~B. Orlin.
\newblock {\em Network {F}lows: theory, algorithms, and applications}.
\newblock Prentice Hall, Englewood Cliffs, New Jersey, 1993.

\bibitem{ABV07}
H.~Aissi, C.~Bazgan, and D.~Vanderpooten.
\newblock Approximation of min-max and min-max regret versions of some
  combinatorial optimization problems.
\newblock {\em European Journal of Operational Research}, 179:281--290, 2007.

\bibitem{ABV09}
H.~Aissi, C.~Bazgan, and D.~Vanderpooten.
\newblock Min-max and min-max regret versions of combinatorial optimization
  problems: a survey.
\newblock {\em European Journal of Operational Research}, 197:427--438, 2009.

\bibitem{AFK02}
S.~Arora, A.~Frieze, and H.~Kaplan.
\newblock A new rounding procedure for the assignment problem with applications
  to dense graph arrangement problems.
\newblock {\em Mathematical Programming}, 92:1--36, 2002.

\bibitem{BN09}
A.~Ben-Tal, L.~El~Ghaoui, and A.~Nemirovski.
\newblock {\em Robust optimization}.
\newblock Princeton Series in Applied Mathematics. Princeton University Press,
  Princeton, NJ, 2009.

\bibitem{BCFM17}
V.~Bil{\`{o}}, I.~Caragiannis, A.~Fanelli, M.~Flammini, and G.~Monaco.
\newblock Simple greedy algorithms for fundamental multidimensional graph
  problems.
\newblock In I.~Chatzigiannakis, P.~Indyk, F.~Kuhn, and A.~Muscholl, editors,
  {\em 44th International Colloquium on Automata, Languages, and Programming,
  {ICALP} 2017}, volume~80 of {\em LIPIcs}, pages 125:1--125:13, 2017.

\bibitem{B12}
C.~B{\"u}sing.
\newblock Recoverable robust shortest path problems.
\newblock {\em Networks}, 59:181--189, 2012.

\bibitem{DW13}
V.~G. Deineko and G.~J. Woeginger.
\newblock Complexity and in-approximability of a selection problem in robust
  optimization.
\newblock {\em 4OR - A Quarterly Journal of Operations Research}, 11:249--252,
  2013.

\bibitem{DK12}
A.~Dolgui and S.~Kovalev.
\newblock Min-max and min-max (relative) regret approaches to representatives
  selection problem.
\newblock {\em 4OR - A Quarterly Journal of Operations Research}, 10:181--192,
  2012.

\bibitem{ER05}
M.~Ehrgott.
\newblock {\em Multicriteria optimization}.
\newblock Springer, 2005.

\bibitem{ER00}
M.~Ehrgott and X.~Gandibleux.
\newblock A survey and annoted bibliography of multiobjective combinatorial
  optimization.
\newblock {\em OR Spectrum}, 22:425--460, 2000.

\bibitem{GRSZ14}
F.~Grandoni, R.~Ravi, M.~Singh, and R.~Zenklusen.
\newblock New approaches to multi-objective optimization.
\newblock {\em Mathematical Programming}, 146:525--554, 2014.

\bibitem{JSWZ21}
S.~Jiang, Z.~Song, O.~Weinstein, and H.~Zhang.
\newblock A faster algorithm for solving general {LP}s.
\newblock In S.~Khuller and V.~V. Williams, editors, {\em {STOC} '21: 53rd
  Annual {ACM} {SIGACT} Symposium on Theory of Computing}, pages 823--832.
  {ACM}, 2021.

\bibitem{KW94}
P.~Kall and S.~W. Wallace.
\newblock {\em Stochastic Programming}.
\newblock John Wiley and Sons, 1994.

\bibitem{KKZ15}
A.~Kasperski, A.~Kurpisz, and P.~Zieli{\'n}ski.
\newblock Approximability of the robust representatives selection problem.
\newblock {\em Operations Research Letters}, 43:16--19, 2015.

\bibitem{KZ09}
A.~Kasperski and P.~Zieli{\'n}ski.
\newblock On the approximability of minmax (regret) network optimization
  problems.
\newblock {\em Information Processing Letters}, 109:262--266, 2009.

\bibitem{KZ16b}
A.~Kasperski and P.~Zieli{\'n}ski.
\newblock Robust {D}iscrete {O}ptimization {U}nder {D}iscrete and {I}nterval
  {U}ncertainty: {A} {S}urvey.
\newblock In {\em Robustness {A}nalysis in {D}ecision {A}iding, {O}ptimization,
  and {A}nalytics}, pages 113--143. Springer-Verlag, 2016.

\bibitem{KY97}
P.~Kouvelis and G.~Yu.
\newblock {\em Robust Discrete Optimization and its Applications}.
\newblock Kluwer Academic Publishers, 1997.

\bibitem{LXZ2024a}
S.~Li, C.~Xu, and R.~Zhang.
\newblock Polylogarithmic approximation for robust s-t path.
\newblock {\em CoRR}, abs/2305.16439, 2023.

\bibitem{LXZ2024}
S.~Li, C.~Xu, and R.~Zhang.
\newblock Polylogarithmic approximations for robust s-t path.
\newblock In K.~Bringmann, M.~Grohe, G.~Puppis, and O.~Svensson, editors, {\em
  51st International Colloquium on Automata, Languages, and Programming,
  {ICALP} 2024}, volume 297 of {\em LIPIcs}, pages 106:1--106:17. Schloss
  Dagstuhl - Leibniz-Zentrum f{\"{u}}r Informatik, 2024.

\bibitem{LLMS09}
C.~Liebchen, M.~E. L{\"{u}}bbecke, R.~H. M{\"{o}}hring, and S.~Stiller.
\newblock The concept of recoverable robustness, linear programming recovery,
  and railway applications.
\newblock In {\em Robust and {O}nline {L}arge-{S}cale {O}ptimization}, volume
  5868 of {\em Lecture Notes in Computer Science}, pages 1--27.
  Springer-Verlag, 2009.

\bibitem{MOT2024}
Y.~Makarychev, M.~Ovsiankin, and E.~Tani.
\newblock Approximation algorithms for $\ell_p$-shortest path and
  $\ell_p$-group steiner tree.
\newblock In K.~Bringmann, M.~Grohe, G.~Puppis, and O.~Svensson, editors, {\em
  51st International Colloquium on Automata, Languages, and Programming,
  {ICALP} 2024}, volume 297 of {\em LIPIcs}, pages 111:1--111:20. Schloss
  Dagstuhl - Leibniz-Zentrum f{\"{u}}r Informatik, 2024.

\bibitem{MOT2024a}
Y.~Makarychev, M.~Ovsiankin, and E.~Tani.
\newblock Approximation algorithms for $\ell_p$-shortest path and
  $\ell_p$-group steiner tree.
\newblock {\em CoRR}, abs/2404.17669, 2024.

\bibitem{PY00}
C.~H. Papadimitriou and M.~Yannakakis.
\newblock On the approximability of trade-offs and optimal access of web
  sources.
\newblock In {\em FOCS}, pages 86--92. IEEE Computer Society, 2000.

\bibitem{RW91}
R.~T. Rockafellar and R.~J.-B. Wets.
\newblock Scenarios and {P}olicy {A}ggregation in {O}ptimization {U}nder
  {U}ncertainty.
\newblock {\em Mathematics of Operations Research}, 16:119--147, 1991.

\bibitem{UT94}
E.~Ulungu and J.~Teghem.
\newblock Multi-objective combinatorial optimization problems: A survey.
\newblock {\em Journal of Multi-criteria Decision Analysis}, pages 83--104,
  1994.

\bibitem{VTL82}
J.~Valdes, R.~E. Tarjan, and E.~L. Lawler.
\newblock The recognition of series parallel digraphs.
\newblock {\em SIAM Journal on Computing}, 11:298--313, 1982.

\bibitem{W89}
R.~J.-B. Wets.
\newblock The aggregation principle in scenario analysis and stochastic
  optimization.
\newblock In S.~W. Wallace, editor, {\em Algorithms and Model Formulations in
  Mathematical Programming}, pages 91--113. Springer-Verlag, 1989.

\bibitem{WS10}
D.~P. Williamson and D.~B. Shmoys.
\newblock {\em The {D}esign of {A}pproximation {A}lgorithms}.
\newblock Cambridge University Press, 2010.

\bibitem{YY98}
G.~Yu and J.~Yang.
\newblock On the robust shortest path problem.
\newblock {\em Computers and Operations Research}, 6:457--468, 1998.

\end{thebibliography}

\appendix

\section{Appendix}
\label{dod}

\begin{proof}[Proof of Proposition~\ref{gapsp}]
Consider an instance of \textsc{Min-Max SP} presented in Figure~\ref{fgapsp}.
We call arcs with zero costs under every scenario \emph{dummy arcs} -- see the dashed ones.
\begin{figure}[ht]
	\centering
	\includegraphics[height=4.5cm]{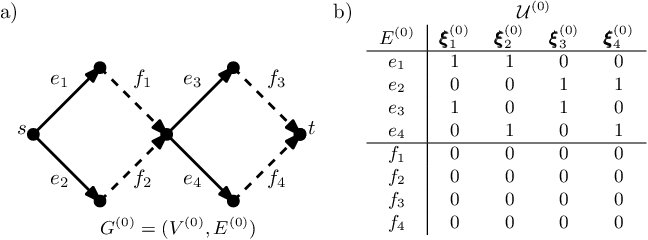}
	\caption{An instance of \textsc{Min-Max SP} with the integrality gap of at least~2.} \label{fgapsp}
\end{figure}
We see at once that $x^*_{e_i}=\frac{1}{2}$, $i\in[4]$, $x^*_{f_i}=\frac{1}{2}$, $i\in[4]$,
is a feasible solution to  $\mathcal{LP}(1)$ (here 
the smallest value of~$C$
for which $\mathcal{LP}(C)$ is feasible is equal to~1, i.e.
$C^*=1$) with the constraints~(\ref{xsp1})-(\ref{xsp2}) and
every integral solution for this instance has the maximum cost over $\mathcal{U}^{(0)}$ equal to~2.
Hence the integrality gap of the linear relaxation~$\mathcal{LP}(C)$ with~(\ref{xsp1})-(\ref{xsp2}), for this instance, is at least~2.
Clearly, graph~$G^{(0)}$  is arc series-parallel.

We now gradually increase the gap. A new instance of \textsc{Min-Max SP}, i.e. graph $G^{(1)}=(V^{(1)}, E^{(1)})$ with
a scenario set~$\mathcal{U}^{(1)}$, is build in the following way.
We replace every arc $e_i$,  $i\in[4]$, in $G^{(0)}$ (see  Figure~\ref{fgapsp}a) by the graph~$G^{(0)}$,
denoted by~$G^{(0)}_{e_i}$, 
obtaining $G^{(1)}$, where
 $|E^{(1)}|=4|E^{(0)}| +4$, $|V^{(1)}|=4(|V^{(0)}|-2) +7$.
 The graph~$G^{(1)}$  is arc series-parallel.
 Then we construct scenario set~$\mathcal{U}^{(1)}$ as follows.
We replace two values of~1  in every scenario $\pmb{\xi}^{(0)}\in \mathcal{U}^{(0)}$ by 
 two matrices $\mathbf{\Xi}^{(1)}_1$ and $\mathbf{\Xi}^{(1)}_2$
of the size $|E^{(0)}|\times |\mathcal{U}^{(0)}|^2$, respectively,
where the columns of the matrices form the Cartesian product  $\mathcal{U}^{(0)}\times \mathcal{U}^{(0)}$, i.e.
\[
\left(
\renewcommand{\arraycolsep}{2pt}
\begin{array}{c}
\mathbf{\Xi}^{(1)}_1\\
\mathbf{\Xi}^{(1)}_2
\end{array}
\right)
 = 
\left(
    \renewcommand{\arraycolsep}{2pt}			
\begin{array}{ccccccccccccc}
	\pmb{\xi}^{(0)}_1 &\pmb{\xi}^{(0)}_1&\ldots &\pmb{\xi}^{(0)}_1  &\pmb{\xi}^{(0)}_2 &\pmb{\xi}^{(0)}_2&\ldots &\pmb{\xi}^{(0)}_2&\ldots    &\pmb{\xi}^{(0)}_{|\mathcal{U}^{(0)}|}&\pmb{\xi}^{(0)}_{|\mathcal{U}^{(0)}|}&\ldots&\pmb{\xi}^{(0)}_{|\mathcal{U}^{(0)}|}\\
	\pmb{\xi}^{(0)}_1  & \pmb{\xi}^{(0)}_2& \ldots & \pmb{\xi}^{(0)}_{|\mathcal{U}^{(0)}|}&\pmb{\xi}^{(0)}_1 & \pmb{\xi}^{(0)}_2& \ldots & \pmb{\xi}^{(0)}_{|\mathcal{U}^{(0)}|}&\ldots&\pmb{\xi}^{(0)}_1&\pmb{\xi}^{(0)}_2&\ldots&\pmb{\xi}^{(0)}_{|\mathcal{U}^{(0)}|}
\end{array}
\right).
\]
 Every value of~0 in $\pmb{\xi}^{(0)}$  that corresponds to arc~$e_i$, $i\in[4]$, 
 is replaced by zero matrix $\mathbf{O}^{(1)}$ of the size  $|E^{(0)}|\times |\mathcal{U}^{(0)}|^2$ and
 every value of~0 in $\pmb{\xi}^{(0)}$  that corresponds to arc~$f_i$, $i\in[4]$, 
 is replaced by zero matrix $\pmb{0}^{(1)}$ of the size  $1\times |\mathcal{U}^{(0)}|^2$
 (see Figure~\ref{fgapsp}b).  Thus $|\mathcal{U}^{(1)}|= |\mathcal{U}^{(0)}|^2 |\mathcal{U}^{(0)}|$.
 The resulting instance is shown in  Figure~\ref{fgapsp1}.
 \begin{figure}[ht]
	\centering
	\includegraphics[height=5cm]{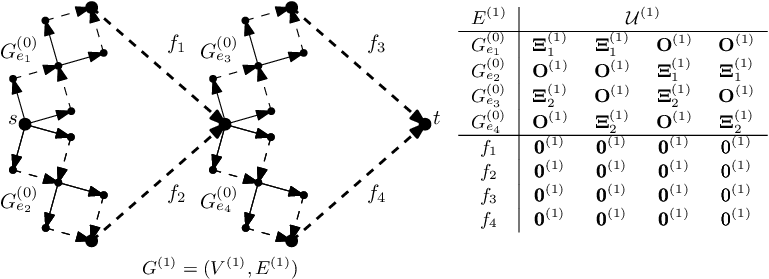}
	\caption{An instance of \textsc{Min-Max SP} with the integrality gap of at least~4} \label{fgapsp1}
\end{figure}
Note that every $s$-$t$ path in $G^{(1)}$ contains exactly four solid arcs. 
From the construction of $\mathcal{U}^{(1)}$, it follows
that there exists a scenario in which the costs of these four arcs are equal to~1. This is the maximum
cost, since each scenario $\pmb{\xi}^{(1)}\in \mathcal{U}^{(1)}$ has exactly four 1's.
Accordingly,  every integral solution for this instance has the maximum cost over $\mathcal{U}^{(1)}$ equal to~4.
Let $\pmb{x}^{*}\in (0,1]^{|E^{(1)}|}$ be given by $x^{*}_{e}=\frac{1}{4}$ for the solid arcs~$e$ in $G^{(1)}$;
$x^*_{f_i}=\frac{1}{2}$ for $i\in[4]$ and 
the components of $\pmb{x}^{*}$ corresponding to the dashed arcs in $G^{(0)}_{e_i}$, $i\in [4]$,
 are equal to~$\frac{1}{4}$. It is easy to check that $\pmb{x}^{*}$ is feasible to $\mathcal{LP}(1)$  ($C^*=1$).
 Therefore,  the integrality gap of the linear relaxation~$\mathcal{LP}(C)$ with~(\ref{xsp1})-(\ref{xsp2}), for this instance, is at least~4.
 
Repeatedly applying the above construction enables us to increase the integrality gap by at least~8.
That is,
we again replace each solid arc~$e_i$,  $i\in[4]$, in $G^{(0)}$ (see  Figure~\ref{fgapsp}a) by the graph~$G^{(1)}$.
This leads to the arc series-parallel graph~$G^{(2)}$ with
 $|E^{(2)}|=4|E^{(1)}| +4$, $|V^{(2)}|=4(|V^{(1)}|-2) +7$.
 Then we build  $\mathcal{U}^{(2)}$ as follows.
We replace two values of~1  in every scenario $\pmb{\xi}^{(0)}\in \mathcal{U}^{(0)}$ by 
 two matrices $\mathbf{\Xi}^{(2)}_1$ and $\mathbf{\Xi}^{(2)}_2$
of the size $|E^{(1)}|\times |\mathcal{U}^{(1)}|^2$, respectively,
where the columns of the matrices form the Cartesian product  $\mathcal{U}^{(1)}\times \mathcal{U}^{(1)}$:
\[
\left(
\renewcommand{\arraycolsep}{2pt}
\begin{array}{c}
\mathbf{\Xi}^{(2)}_1\\
\mathbf{\Xi}^{(2)}_2
\end{array}
\right)
 = 
\left(
    \renewcommand{\arraycolsep}{2pt}			
\begin{array}{ccccccccccccc}
	\pmb{\xi}^{(1)}_1 &\pmb{\xi}^{(1)}_1&\ldots &\pmb{\xi}^{(1)}_1  &\pmb{\xi}^{(1)}_2 &\pmb{\xi}^{(1)}_2&\ldots &\pmb{\xi}^{(1)}_2&\ldots    &\pmb{\xi}^{(1)}_{|\mathcal{U}^{(1)}|}&\pmb{\xi}^{(1)}_{|\mathcal{U}^{(1)}|}&\ldots&\pmb{\xi}^{(1)}_{|\mathcal{U}^{(1)}|}\\
	\pmb{\xi}^{(1)}_1  & \pmb{\xi}^{(1)}_2& \ldots & \pmb{\xi}^{(1)}_{|\mathcal{U}^{(1)}|}&\pmb{\xi}^{(1)}_1 & \pmb{\xi}^{(1)}_2& \ldots & \pmb{\xi}^{(1)}_{|\mathcal{U}^{(1)}|}&\ldots&\pmb{\xi}^{(1)}_1&\pmb{\xi}^{(1)}_2&\ldots&\pmb{\xi}^{(1)}_{|\mathcal{U}^{(1)}|}
\end{array}
\right), \; \pmb{\xi}^{(1)}\in \mathcal{U}^{(1)}.
\]
 Every~0 in $\pmb{\xi}^{(0)}$  that corresponds to arc~$e_i$, $i\in[4]$, 
 is replaced by matrix $\mathbf{O}^{(2)}$ of the size  $|E^{(1)}|\times |\mathcal{U}^{(1)}|^2$ and
 every~0 in $\pmb{\xi}^{(0)}$  that corresponds to arc~$f_i$, $i\in[4]$, 
 is replaced by matrix~$\pmb{0}^{(2)}$ of the size  $1\times |\mathcal{U}^{(1)}|^2$ 
 and so
 $|\mathcal{U}^{(2)}|= |\mathcal{U}^{(1)}|^2 |\mathcal{U}^{(0)}|$.
 
The proof of the proposition is based on repeating the above construction~$r$ times.
We get  the arc series-parallel graph~$G^{(r)}=(V^{(r)},E^{(r)})$, where 
 $|E^{(r)}|=4|E^{(r-1)}| +4$, $|V^{(r)}|=4(|V^{(r-1)}|-2) +7$, and 
 scenario set~$\mathcal{U}^{(r)}$ with the cardinality of $|\mathcal{U}^{(r-1)}|^2 |\mathcal{U}^{(0)}|$.
 Clearly, the resulting graph is arc series-parallel.
 Now, the integrality gap of the linear relaxation~$\mathcal{LP}(L)$ with~(\ref{xsp1})-(\ref{xsp2}) for the resulting
 instance is at least~$2^{r+1}$. Set $k=2^{r+1}$, $K=|\mathcal{U}^{(r)}|$, $m=|E^{(r)}|$ and $n=|V^{(r)}|$.
 Since $|\mathcal{U}^{(0)}|=4$, $|E^{(0)}|=8$ and $|V^{(0)}|=7$,
 a trivial verification yields   $K=4^{2^{r+1}-1}=4^{k-1}$, 
 $m=4^{r+1}+\sum_{p=1}^{r+1}4^p=\frac{7}{3}k^2-\frac{4}{3}$
 and $n=2+5\sum_{p=0}^{r} 4^p=2+\frac{5}{3}(k^2-1)$. 
 The latter equation shows that the linear relaxation~$\mathcal{LP}(C)$ with~(\ref{xsp1})-(\ref{xsp2})
has an integrality gap of at least $\Omega(\sqrt{n})$, which completes the proof.
\end{proof}

\end{document}